\documentclass[sigconf,9pt]{acmart}

\AtBeginDocument{%
  \providecommand\BibTeX{{%
    \normalfont B\kern-0.5em{\scshape i\kern-0.25em b}\kern-0.8em\TeX}}}

\copyrightyear{2020} 
\acmYear{2020} 
\setcopyright{acmcopyright}
\acmConference[HSCC '20]{23rd ACM International Conference on Hybrid Systems: Computation and Control}{April 22--24, 2020}{Sydney, NSW, Australia}
\acmBooktitle{23rd ACM International Conference on Hybrid Systems: Computation and Control (HSCC '20), April 22--24, 2020, Sydney, NSW, Australia}
\acmPrice{15.00}
\acmDOI{10.1145/3365365.3382212}
\acmISBN{978-1-4503-7018-9/20/04}

\usepackage{enumitem}
\usepackage{tikz}
\usetikzlibrary{automata,arrows,positioning,calc}

\newtheorem{definition}{Definition}

\newtheorem{problem}{Problem}
\newtheorem{lemma}{Lemma}
\newtheorem{remark}{Remark}

\newtheorem{theorem}{Theorem}

\DeclareMathOperator{\dist}{{\text{dist}}}

\begin{document}

\title[A Convex Parameterization of Assume-Guarantee Contracts for Linear Systems]{Compositional Synthesis via a Convex Parameterization of Assume-Guarantee Contracts}
\author{Kasra Ghasemi}
\affiliation{%
  \institution{Boston University}
  \city{Boston, MA}
  \country{USA \\ kasra0gh@bu.edu}
}

\author{Sadra Sadraddini}
\affiliation{%
  \institution{Massachusetts Institute of Technology}
  \city{Cambridge, MA}
  \country{USA \\ sadra@mit.edu}
}

\author{Calin Belta}
\affiliation{%
  \institution{Boston University}
  \city{Boston, MA}
  \country{USA \\ cbelta@bu.edu}
}

\renewcommand{\shortauthors}{Ghasemi. Sadraddini, Belta}

\begin{abstract}
We develop an assume-guarantee framework for control of large scale linear (time-varying) systems from finite-time reach and avoid or infinite-time invariance specifications. The contracts describe the admissible set of states and controls for individual subsystems. A set of contracts compose correctly if mutual assumptions and guarantees match in a way that we formalize.  We propose a rich parameterization of contracts such that the set of parameters that compose correctly is convex. Moreover, we design a potential function of parameters that describes the distance of contracts from a correct composition. Thus, the verification and synthesis for the aggregate system are broken to solving small convex programs for individual subsystems, where correctness is ultimately achieved in a compositional way. Illustrative examples demonstrate the scalability of our method.     
\end{abstract}





\maketitle





\section{Introduction}

Formal verification and synthesis are computationally expensive and traditional methods fail in large-scale systems. Thus, approaches that exploit inherent modular structures have been proposed to deal with the scalability issue.  
Such structures are present in a wide variety of cyber-physical systems such as energy management \cite{bowen1993formal,clarke2008analysis}, transportation \cite{clarke2008analysis,coogan2015compartmental}, and biological engineering \cite{batt2007robustness,nielsen2016genetic}.  

Assume-guarantee reasoning (AGR) \cite{henzinger2001assume,giannakopoulou2004assume} is a divide and conquer approach to verification and control. AGR involves contracts, which, in plain words, describe the promises that individual modules of an aggregate system take from and make to the environment, which are passed to other subsystems in a circular or hierarchical fashion. If these promises are carefully designed, the verification/control of the aggregate system is achieved in a scalable way. AGR originates from the formal methods community, where the main application domain has traditionally been discrete space systems in software model checking \cite{clarke1999model,baier2008principles}. However, AGR for continuous space and hybrid models in engineering applications has become an active research area in recent years \cite{Nilsson2016seperable,kim2017symbolic,sadraddini2017formal_traffic,kim2017dynamic,oh2019optimizing,nuzzo2019stochastic,saoud2019assume}. 

While there are important works on the theoretical foundations of AGR and how contracts should be used \cite{kwiatkowska2010assume,kerber2010compositional,nuzzo2019stochastic}, there is little work on how to find the contracts themselves. Designing contracts, especially for circular reasoning, is a much harder problem than using a set of given contracts. A relevant problem is \emph{assumption mining} \cite{li2011mining}, which aims to represent a subset of environment assumptions that lead to desirable system behaviors, typically described by temporal logics \cite{baier2008principles}. Assumption mining is non-trivial for control systems operating in continuous space. The contracts may represent the admissible set of disturbances/couplings for individual subsystems, but there is no satisfying answer on how to search for such sets. A natural approach is parameterization and searching for those that facilitate compositional verification and control \cite{kim2017small}. 

In this paper, we develop a parametric assume-guarantee approach for a network of discrete-time linear systems with (weakly) coupled dynamics and affected by disturbances. We consider both finite-time reach and avoid specifications and infinite-time set-invariance. The goal is verification and control in a fully compositional way such that instead of directly solving the intractable large problem, we solve small problems corresponding to each subsystem multiple times. Our AGR can be circular - a harder task than cascade or hierarchical reasoning - as every subsystem may interact with every other subsystem. The synthesis problem involves finding the contract parameters and the controllers at the same time. The resulting controllers are correct-by-design local state feedback. 

\begin{figure*}[t] 
  \centering
  \includegraphics[width=0.48\textwidth]{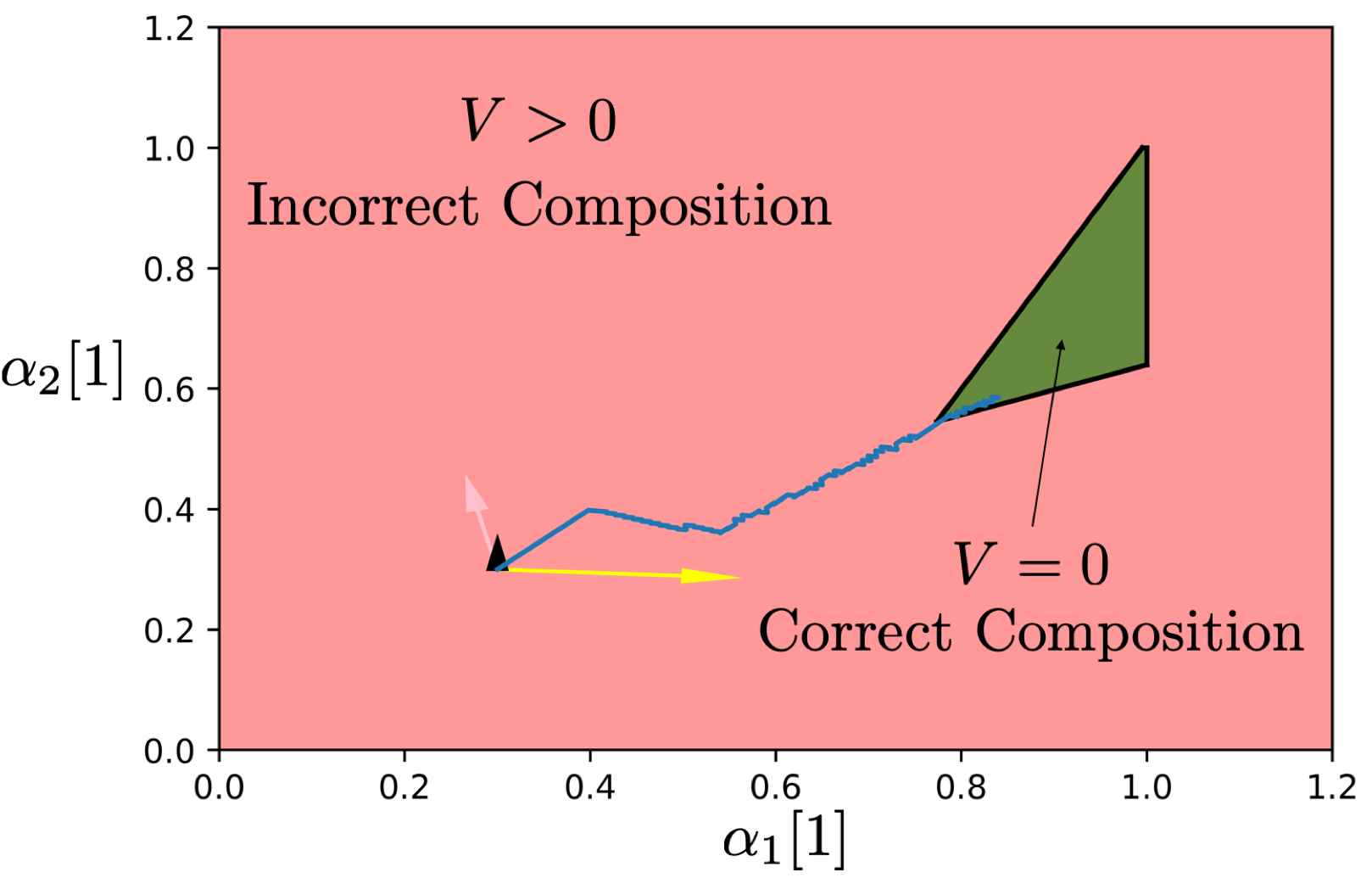}
  \includegraphics[width=0.49\textwidth]{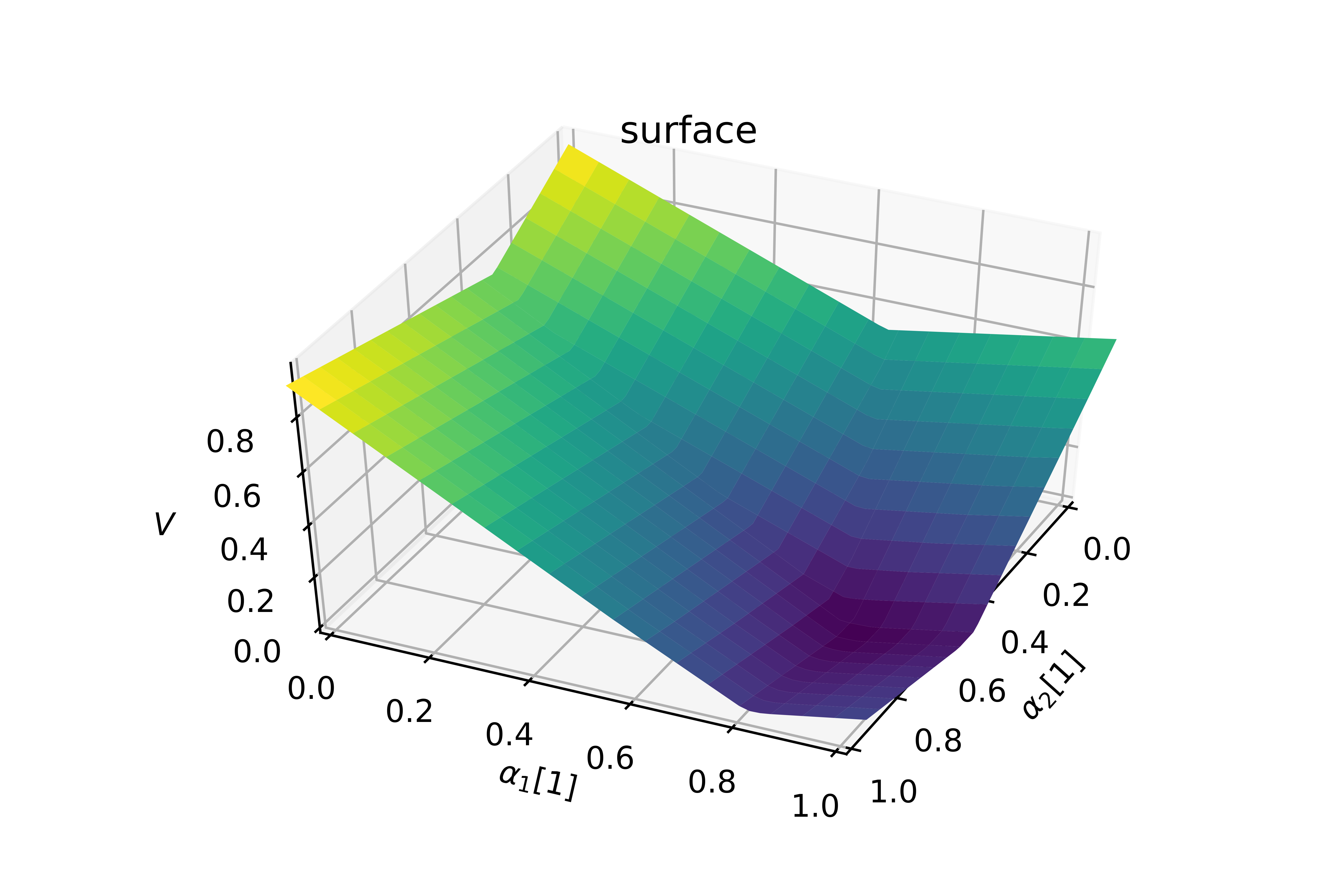}
  \caption{[Left] The green polytope is the projection of the correct composition parameters for three 2D subsystems, a total of 6 parameters. Given an initial guess of the parameters, we use the dual solutions of specific convex programs (explained in Section \ref{sec_gradient}) for each subsystem to obtain gradient information of the potential function, which is the summation of gradients corresponding to each subsystem. The gradient descent parameters toward a correct composition. [Right] The potential function with respect to two parameters, when other parameters are fixed. The values $V=0$ correspond to correct composition.}
  \label{figurelabe2}
\end{figure*}

\subsection{Contributions and Organization}
We provide the necessary background in Section \ref{sec_prelim}, and formalize the problem in Section \ref{sec_problem}. Contributions of this paper are as follows:
\begin{enumerate}[leftmargin=*]
\item We introduce a framework to characterize the assume-guarantee pair for an individual subsystem. We show how to cast the computations of finite-time and infinite-time viable sets represented by zonotopes \cite{ziegler2012lectures} using convex linear programs (Section \ref{sec_AG}). 
\item We introduce a specific form of parametric contracts and define the notion of correctness for composing a set of contracts for the whole system. We show that the set of parameters leading to correct composition is convex. Furthermore, we introduce a potential function of parameters that characterize the distance from correctness. This potential function is convex (e.g. Fig \ref{figurelabe2}[Bottom]) and is a summation of directed Hausdorff distances \cite{rockafellar2015convex,sadraddini2019linear}, each defined for an individual subsystem (Section \ref{sec_AGR}).
\item We compute the correct contract parameters. While a single convex program can find the (optimal) correct parameters, we show that the same is achieved by solving smaller convex programs for each subsystem and summing the information from each dual solution to find the gradient of the potential function (e.g. Fig \ref{figurelabe2}[Top]) (Section \ref{sec_gradient}). 
\item We present illustrative examples and numerical benchmarks on the usefulness of our method and demonstrate its scalability to very large problems (Section \ref{sec_examples}). 
\end{enumerate}

\subsection{Related Work}
\subsubsection*{Monotone systems} Motivated by vehicular transportation models, a particular class of problems that compositional synthesis has had some success is monotone systems, where the dynamics and specifications preserve specific forms of partial order relations. It was shown in \cite{kim2016directed,kim2017symbolic} that assumption mining for this class of problem can be cast as a multi-dimensional binary search. In \cite{kim2015compositional}, the assumptions were fixed sets of coupling states, and they were extended to periodic sets in \cite{sadraddini2017formal_traffic} and dynamic ones in \cite{kim2017dynamic}. Monotonicity does not hold, in general, for constrained linear systems such as those considered in this paper. 

\subsubsection*{Linear systems} The closest works are \cite{Nilsson2016seperable} and \cite{ghasemi2019compositional}. In \cite{Nilsson2016seperable}, the constrained set-invariance problem for discrete-time disturbed linear systems was considered. However, the synthesis method was \emph{not} compositional as the parameters of the controllers and the sets corresponding to assume-guarantee contracts, also characterized by zonotopes, were computed by solving a single large linear matrix inequality (LMI) problem. In \cite{ghasemi2019compositional}, the same problem as in \cite{Nilsson2016seperable} was considered, the method was compositional, and was shown to scale well to 1000+ dimensions. However, there was no contract parameterization. Instead, the method iteratively searched for a limited set of contracts and the problem would not have been solved if the first guess of the contracts was not appropriate. 

\subsubsection*{Finite-state systems} The authors in \cite{Eqtami2019} presented a compositional method to compute assume-guarantee contracts. They introduced a parametric approach with a correctness criterion to deal with the circularity issue of AGR. The method is originally designed for finite transition systems. Extensions to infinite transition are done through finite-state abstractions. However, abstraction methods are themselves, in general, computationally expensive and can cause considerable conservatism. Additionally, the parameterization is not as rich as the one in this paper as one scalar parameter is attributed to each subsystem. Furthermore, there is no convergence guarantee for the parameters. Finally, the contract search is performed in an exhaustive fashion over the finite states, which hinders scalability.  



\section{Notation and Preliminaries}
\label{sec_prelim}
The set of real numbers and non-negative integers are denoted by $\mathbb{R}$ and $\mathbb{N}$, respectively. We denote the non-negative integers not larger than $h$ by $\mathbb{N}_h$. Given two matrices $A_1$ and $A_2$ with the same number of rows, $[A_1,A_2]$ denotes their horizontal stacking. Given a vector $\alpha \in \mathbb{R}^n$, $\text{Diag}(\alpha)$ is a $n \times n$ diagonal matrix where $\alpha$ is its main diagonal and sum($\alpha$) is the sum of the elements of $\alpha$. We denote the $n\times n$ identity matrix by $I_n$, where $n \in \mathbb{N}$. The logical operation ``and'' is denoted by $\wedge$. Given two sets $S_1, S_2 \subset \mathbb{R}^n$, their Minkowski sum is denoted by $S_1 \oplus S_2 = \{s_1+s_2| s_1 \in S_1, s_2 \in S_2\}$. Given $s \in \mathbb{R}^n, S \subset \mathbb{R}^n$ and $T \in \mathbb{R}^{q \times n}$, we interpret $s+S$ as $\{s\} \oplus S$ and $TS=\{Ts|s \in S\}$.

A zonotope is a symmetric set represented as $\mathcal{Z}(x_c,G) := x_c \oplus G \mathbb{B}_p \subset \mathbb{R}^n$, where $x_c \in \mathbb{R}^n$ is the zonotope centroid, the columns of $G \in \mathbb{R}^{n \times p}$ are the zonotope generators, and $\mathbb{B}_p:=\{x\in \mathbb{R}^p| ||x||_\infty \leq 1\}$. Also, the order of a zonotope is defined as $\dfrac{p}{n}$. 
Zonotopes are convenient to manipulate with affine transformations and Minkowski sums \cite{kopetzki2017methods}, as shown below:
\begin{subequations}
\begin{equation}
\label{eq_zonotope_affine}
    A \mathcal{Z}(\bar{x},G)+b=\mathcal{Z}(A\bar{x} + b,AG)
\end{equation}
\begin{equation}
\label{eq_zonotope_minkowski}
    \mathcal{Z}(\bar{x}_1,G_1) \oplus \mathcal{Z}(\bar{x}_2,G_2) = \mathcal{Z}(\bar{x}_1+\bar{x}_2,[G_1,G_2]).
\end{equation}
\end{subequations}

The \textit{Directed Hausdorff distance} of two sets, denoted by $d_{DH}(\mathbb{S}_1,\mathbb{S}_2)$ is a quantitative measure on how distant is the set $\mathbb{S}_2 \subset \mathbb{R}^n$ from being a subset of $\mathbb{S}_1 \subset \mathbb{R}^n$:
\begin{equation}
    d_{DH}(\mathbb{S}_1,\mathbb{S}_2) =  \sup_{s_2\in \mathbb{S}_2} \inf_{s_1\in \mathbb{S}_1} d(s_1,s_2) ,
\end{equation}
where $d(s_1,s_2)$ is a metric between points $s_1$ and $s_2$. We use infinity-norm in this paper. For closed compact sets, we have $d_{DH}(\mathbb{S}_1,\mathbb{S}_2)=0$ if and only if $\mathbb{S}_2 \subseteq \mathbb{S}_1$. Given a set $Y \subseteq X$, and a function $\mu: X \rightarrow U$, we interpret $\mu(Y)\subseteq U$ as $\{\mu(y)|y \in Y\}$.

\section{Problem Statement}
\label{sec_problem}

Consider the following discrete-time time-varying linear system:
\begin{equation}\label{singlesystem_LTV}
    x(t+1) = A(t)x(t)+B(t)u(t)+d(t),
\end{equation}
where $x(t)\in X(t)$, $u(t)\in U(t)$, and $d(t)\in D(t)$. The sets $X(t), D(t) \subset \mathbb{R}^n$, and $U(t) \subset \mathbb{R}^m$ define time varying constraints over the state, disturbance, and control input, respectively. The matrices $A(t) \in \mathbb{R}^{n\times n}$ and $B(t) \in \mathbb{R}^{n \times m}$ may be time dependent, and $t \in \mathbb{N}$ is time.

When the matrices $A(t)$ and $B(t)$ and the bounds $X(t) , U(t)$, and $D(t)$ are time-invariant and we are interested in the infinite-time response of the system, we consider the system to be linear time-invariant (LTI). Throughout the paper, the notation $(t)$ refers to linear time-variant (LTV) class of problems.

A control policy $\mu$ is a set of functions $\mu(t): X(t) \rightarrow U(t), t \in \mathbb{N}_h$. For infinite horizon, the policy $\mu: X \rightarrow U$ is memoryless.

\begin{definition}[Finite-time Viable Sets] \label{def:viable}
A sequence of sets $\{\Omega(t) | \Omega(t) \subseteq X(t), t \in \mathbb{N}_h\}$ for  system \eqref{singlesystem_LTV} is \emph{viable} if there exists a policy $\mu$ such that $\Theta(t) \subseteq U(t), \Theta(t)=\mu(t)(\Omega(t))$, and 
\begin{multline} \label{eq:4}
    \forall t \in \mathbb{N}_{h-1}, \forall x(t)\in \Omega(t), \forall d(t)\in D(t)\\ \Rightarrow x(t+1)\in \Omega(t+1).
\end{multline}
\end{definition}
\begin{definition}[Infinite-time Viable set] \label{def:RCI}
A set $\Omega \subseteq X$ for an LTI system is infinite-time viable, also known as robust control invariant, if there exists control policy $\mu$ s.t. $\Theta=\mu(\Omega) \subseteq U$, and 
\begin{multline}
    \forall t\in \mathbb{N}, \forall x(t)\in \Omega , \exists u(t)\in U, \forall d(t) \in D \Rightarrow x(t+1) \in \Omega.
\end{multline}
\end{definition}


In this paper, we deal with networks of coupled linear systems of the following form:  
\begin{multline}
\label{subsystems_LTV}
    x_i(t+1)=A_{ii}(t)x_i(t) + B_{ii}(t)u_i(t) + \sum_{j\ne i}{A_{ij}(t)x_j(t)} \\ + \sum_{j\ne i}{B_{ij}(t)u_j(t)} + d_i(t),
\end{multline}
where $x_i(t)\in X_i(t), X_i(t) \subseteq  \mathbb{R}^{n_i}$, $u_i(t) \in U_i(t), U_i(t) \subseteq  \mathbb{R}^{m_i}$, and $d_i(t)\in D_i(t), D_i(t) \subseteq \mathbb{R}^{n_i}$ are the state, control input, and disturbance for subsystem $i$, respectively. The (time-varying) matrices $A_{ii}(t) \in \mathbb{R}^{n_i \times n_i}$ and $B_{ii}(t) \in \mathbb{R}^{n_i \times m_i}$ characterize the internal dynamics of subsystem $i$. Also, $A_{ij}(t) \in \mathbb{R}^{n_i \times n_j}$ and $B_{ij}(t) \in \mathbb{R}^{n_i \times m_j}$ characterizes the coupling effects of subsystem $j$ on subsystem $i$. The index of a subsystem is shown by $i \in \mathcal{I}$, where $\mathcal{I}$ is the index set for subsystems and $t \in \mathbb{N}_h$ is the time step. 

We assume time-variant sets $X_i(t),U_i(t)$, and $D_i(t)$ are zonotopes $\mathcal{Z}(\bar{x}_i(t), G_i^{x}(t))$, $\mathcal{Z}(\bar{u}_i(t), G_i^{u}(t))$, and $\mathcal{Z}(\bar{d}_i(t), G_i^{d}(t))$, respectively, where the vectors $\bar{x}_i(t),\bar{d}_i(t) \in \mathbb{R}^{n_i}$, and $\bar{u}_i(t) \in  \mathbb{R}^{m_i}$ and matrices $G_i^x(t) \in \mathbb{R}^{n_i \times p_i^x(t)}$, $G_i^u(t) \in \mathbb{R}^{m_i \times p_i^u(t)}$, and $G_i^d(t) \in \mathbb{R}^{n_i \times p_i^d(t)}$ are given. Note that these assumptions are not restrictive, since zonotopes can tightly under/over approximate symmetric shapes. 

The controller for each subsystem has access only to its own state information:
\begin{equation}
\mu_i(t): X_i(t) \rightarrow U_i(t), \text{~(Finite Horizon)}.    
\end{equation} 
\begin{equation}
\mu_i: X_i \rightarrow U_i, \text{~(Infinite Horizon)}.    
\end{equation} 
The first problem is finding these decentralized controllers that lead to viable sets. Decentralized controllers have the advantage that do not require communications in the networked linear system. 

\begin{problem}[Decentralized Finite-time Viable Sets] \label{problem_viable}
Given system \eqref{subsystems_LTV}, find sets $\Omega_i(t)$ and controllers $\mu_i(t), i \in \mathcal{I}$ and $t \in \mathbb{N}_h$ such that $\Omega_i(t) \subseteq X_i(t)$, $\Theta_i(t) \subseteq U_i(t)$, and
\begin{multline} \label{eq:subrci}
        \forall x_i(t) \in \Omega_i(t), \forall x_j(t)\in \Omega_j(t),
        \forall u_j(t)\in \Theta_j(t), (j\ne i) , \\ \forall d_i(t) \in D_i(t) \Rightarrow x_i(t+1)\in \Omega_i(t+1),
\end{multline}
where $\Theta_i(t)=\mu_i(t)(\Omega_i(t))$.
\end{problem}


The second problem is finding decentralized infinite-time viable sets or robust set-invariance controllers for each subsystem.

\begin{problem}[Decentralized Infinite-time Viable Sets] \label{Problem_RCI}
Given each subsystem in time-invariant form of (\ref{subsystems_LTV}) (drop ($t$)), find sets $\Omega_i$ and $\mu_i$ for all $i \in \mathcal{I}$, such that $\Omega_i \subseteq X_i$, $\Theta_i \subseteq U_i$, and
\begin{multline} \label{RCI_prob1}
        \forall x_i(t) \in \Omega_i , \forall x_j(t)\in \Omega_j , \forall u_j(t)\in \Theta_j (j\ne i) \\ , \forall d_i(t) \in D_i \Rightarrow x_i(t+1)\in \Omega_i,
\end{multline}
where $\Theta_i= \mu_i(\Omega_i)$. 
\end{problem}
Note that the optimality criteria can be added to both Problem \ref{problem_viable} and \ref{Problem_RCI}. We also note that the concept of infinite-time viable sets and Problem \ref{Problem_RCI} can be extended to $T$-periodic systems where $A(t+T)=A(t), B(t+T) = B(t), X(t+T)=X(t), U(t+T)=U(t), D(t+T)=D(t), \forall t \in \mathbb{N}$. We omit studying this class of systems in this paper.

\section{Assume-Guarantee contracts}
\label{sec_AG}

In this section, we formalize assume-guarantee contracts for a single system and provide details on the convex parameterization of contracts and controllers.  

\subsection{Definitions}
\begin{definition}[Assume-Guarantee contract]
An assume-guarantee contract for system \eqref{singlesystem_LTV} is a pair $\mathcal{C} = (\mathcal{A},\mathcal{G})$, where:
\begin{itemize}
    \item $\mathcal{A}$ is the assumption, which is the  sequence of disturbance sets $\mathcal{D}(t), t \in \mathbb{N}_h$ (finite-time LTV), or the static set $\mathcal{D}$ (infinite-time LTI) ;
    \item $\mathcal{G}$ is the guarantee, which is  the sequence of admissible states  $\mathcal{X}(t)$ and control inputs $\mathcal{U}(t)$ (finite-time LTV), or static sets $\mathcal{X}, \mathcal{U}$ (infinite-time LTI).  
\end{itemize}
\end{definition}

\begin{definition}
A contract is \emph{valid} if its guarantees respect the system constraints $\mathcal{X}(t) \subseteq X(t), \forall t \in \mathbb{N}_h, \mathcal{U}(t) \subseteq U(t), \forall t \in \mathbb{N}_{h-1}$ (finite horizon), or $\mathcal{X} \subseteq X, \mathcal{U} \subseteq U$ (infinite horizon).  
\end{definition}

\begin{definition}
A valid contract is \emph{satisfiable} if it is possible to find a control policy and viable sets such that $\Omega(t) \subseteq \mathcal{X}(t), \forall t \in \mathbb{N}_h, \Theta(t) \subseteq \mathcal{U}(t), \forall t \in \mathbb{N}_{h-1}$ (finite horizon), or $\Omega \subseteq \mathcal{X}, \Theta \subseteq \mathcal{U}$ (infinite horizon).
\end{definition}

We show that the satisfiability of contracts can be checked using convex programs, which encode a specific form of control policies. 

\subsection{Finite Horizon Contract Satisfiability}
\begin{theorem} \label{Thrm_viable set_single} 
Given a system in the form (\ref{singlesystem_LTV}), a finite horizon contract is satisfiable, if
 $\exists k \in \mathbb{N}$, vectors $\bar{\mathrm{x}}(t)\in \mathbb{R}^n , \bar{\mathrm{u}}(t) \in \mathbb{R}^m$, and matrices $ T(t)\in \mathbb{R}^{n \times l(t)}$ and $M(t)\in \mathbb{R}^{m \times l(t)}$ where $l(0)=k$ and $l(t \ne 0) = k + \sum_{t=0}^{t-1}{p(t)}$ such that the following relation holds:\\
\begin{subequations} \label{singleviablecon}
\begin{equation} \label{single_viable_constraint}
    [A(t)T(t)+B(t)M(t) , G^d(t)] = T(t+1), \forall t \in \mathbb{N}_{h-1},
\end{equation}
\begin{equation}
    A(t) \bar{\mathrm{x}}(t)+ B(t) \bar{\mathrm{u}}(t) + \bar{d}(t) = \bar{\mathrm{x}}_{t+1}, \forall t \in \mathbb{N}_{h-1},
\end{equation}
\begin{equation} \label{singleviable_statecon}
    \mathcal{Z}(\bar{\mathrm{x}}(t),T(t)) \subseteq X(t), \forall t \in \mathbb{N}_{h},
\end{equation}
\begin{equation} \label{singleviable_controlcon}
    \mathcal{Z}(\bar{u}(t),M(t)) \subseteq U(t), \forall t \in \mathbb{N}_{h-1}.
\end{equation} 
\end{subequations}
then, $ \Omega(t)=\mathcal{Z}(\bar{\mathrm{x}}(t),T(t)) , \forall t\in \mathbb{N}_h$ is a sequence of viable sets for horizon $h$. Moreover, the controller is given by the following relation:
\begin{equation}
\label{eq_control_finite}
\mu(t)(x)=\bar{\mathrm{u}}(t)+M(t) \zeta(x), x=\bar{\mathrm{x}}(t)+T(t)\zeta(x), \zeta \in \mathbb{B}_{l(t)},
\end{equation}
and $\Theta(t)=\mathcal{Z}(\bar{\mathrm{u}}(t),M(t))$.
\end{theorem}
\begin{proof}
The proof is by construction. Substituting \eqref{eq_control_finite} in (\ref{singlesystem_LTV}) yields:
\begin{multline*}
    x(t+1)     \in   A(t)(T(t) b + \bar{\mathrm{x}}(t))+B(t)(M(t) b + \bar{\mathrm{u}}(t) ) + \bar{d}(t) \\ \oplus G^d(t) \mathbb{B}_{p(t)}, 
\end{multline*}
where the right hand side set is subset of the following set:
\begin{multline*}
    \{A(t) \bar{\mathrm{x}}(t)+ B(t) \bar{\mathrm{u}}(t) + \bar{d}(t)\} \\ \oplus [A(t)T(t)+B(t)M(t),G^d(t)] \mathbb{B}_{l(t)+p(t)},
\end{multline*}
by using the Minkowski sum property of zonotopes \eqref{eq_zonotope_minkowski}, it is straightforward to reach to \eqref{singleviablecon}. 

\end{proof}

Note that the structure of matrices $T$ and $M$ is not unique and the value of $k$ can be changed to derive different $T(t)$ and $M(t)$. This enables iterations over different $k$.\\

\begin{lemma}[Zonotope Containment \cite{sadraddini2019linear}]
\label{sadra_zon_containment}
Given two zonotopes $\mathcal{Z}(c_1,G_1)$ and $\mathcal{Z}(c_2,G_2)$, where $c_1 ,c_2 \in \mathbb{R}^q$ and $G_1 \in \mathbb{R}^{q \times r}$, $G_2 \in \mathbb{R}^{q \times s}$, we have $\mathcal{Z}(c_1,G_1) \subseteq \mathcal{Z}(c_2,G_2)$ if $\exists \Gamma \in \mathbb{R}^{s \times r}$ and $\gamma \in \mathbb{R}^s$ s.t.
\begin{subequations} \label{zon_containment}
\begin{equation}
    G_1 = G_2 \Gamma, 
\end{equation}
\begin{equation}
     c_2 - c_1 = G_2 \gamma, 
\end{equation}
\begin{equation} \label{con_zon_containment}
    || [\Gamma,\gamma] ||_\infty \leq 1.
\end{equation}
\end{subequations}
\end{lemma}
While Lemma \ref{sadra_zon_containment} states a sufficiency condition, it was shown in \cite{sadraddini2019linear} that its necessity gap is often very small.

Using the Lemma \ref{sadra_zon_containment}, the constraints in \eqref{singleviable_statecon} and \eqref{singleviable_controlcon} become linear in $T(t), M(t), \bar{x}(t),$ and $ \bar{u}(t)$. Therefore, one can check satisfiability of contracts using convex linear programs. The cost function is ad-hoc. We typically choose to minimize the summation of Frobenious norms of $T(t)$ for $t\in \mathbb{N}_h$. 

\begin{remark}
Note that the order of zonotope $\Omega(t)$ is increasing at each time step. This makes the number of variables and constraints in the program grow quadratically with $h$. We can decrease the complexity by fixing the number of columns of $T(t)$ and $M(t)$ to $k$ and change the equation (\ref{single_viable_constraint}) to
\begin{equation}
    [A(t)T(t)+ B(t)M(t) , G(t)^d] =  [0_{n \times p(t)} , T_{t+1}].
\end{equation}
However, this modification leads to a more conservative computation and may cause infeasibility.
\end{remark}

\subsection{Infinite Horizon Contract Satisfiability}
Inspired by the method in \cite{rakovic2007optimized}, we provide a linear programming approach to compute robust control invariant sets. 
\begin{theorem} 
\label{thrm_single_RCI}
An infinite horizon contract is satisfiable, if
 $\exists k \in \mathbb{N}, \beta \in [0,1)$, vectors $\bar{\mathrm{x}}\in \mathbb{R}^n , \bar{\mathrm{u}} \in \mathbb{R}^m$, and matrices $ T\in \mathbb{R}^{n \times k}$ and $M\in \mathbb{R}^{m \times k}$, such that the following relation holds
\begin{subequations}
\begin{equation} \label{eq:thrm1}
    [AT+BM , G^d] = [ E , T],
\end{equation}
\begin{equation} \label{Econ}
    \mathcal{Z}(0,E) \subseteq \mathcal{Z}(0,\beta G^d),
\end{equation}
\begin{equation} \label{eq:thrm2}
    \dfrac{1}{1-\beta} \mathcal{Z}(\bar{\mathrm{x}},T) \subseteq X,
\end{equation}
\begin{equation} \label{eq:thrm3}
    \dfrac{1}{1-\beta} \mathcal{Z}(\bar{\mathrm{u}},M) \subseteq U,
\end{equation} 
\begin{equation}
    A\bar{\mathrm{x}} + B \bar{\mathrm{u}} + \bar{d} = \bar{\mathrm{x}},
\end{equation}
\end{subequations}
then $\Omega=\mathcal{Z}(\bar{\mathrm{x}},(1-\beta)^{-1}T)$ is a robust control invariant set. Furthermore, the controller is given by
\begin{equation}
\label{eq_control_infinite}
\mu(x)=\bar{\mathrm{u}}+M \zeta(x), x=\bar{\mathrm{x}}+T\zeta(x), \zeta \in \mathbb{B}_k,
\end{equation}
and $\Theta=\dfrac{1}{1-\beta}\mathcal{Z}(\bar{\mathrm{u}},M)$.
\end{theorem}

\begin{proof}
Substituting policy \eqref{eq_control_infinite} in \eqref{singlesystem_LTV}, we obtain the relation \eqref{eq:thrm1}. In order to prove invariance, we observe that:
\begin{equation}
    (AT+BM)\mathbb{B}_k \oplus \mathcal{Z}(0,G^d) \subseteq T\mathbb{B}_k \oplus \mathcal{Z}(0,\beta G^d). 
\end{equation}
We substract $\mathcal{Z}(0,\beta G^d)$ from both sides in Pontryagin difference sense \cite{kolmanovsky1998theory}, which we claim that is a valid operation when both sides are convex polytopes. We omit the proof as it is based on the properties of support functions \cite{rockafellar2015convex} of convex sets. We arrive in: 
\begin{equation}
\label{eq_middle_sadra}
    (AT+BM)\mathbb{B}_k \oplus (1-\beta) \mathcal{Z}(0,G^d) \subseteq T\mathbb{B}_k. 
\end{equation}
By multiplying both sides of \eqref{eq_middle_sadra} by $\dfrac{1}{1-\beta}$ we arrive in the conclusion that $\dfrac{1}{1-\beta} \mathcal{Z}(\bar{\mathrm{x}},T)$ is a robust control invariant set with $\Theta=\dfrac{1}{1-\beta}\mathcal{Z}(\bar{\mathrm{u}},M)$, and the proof is complete.
\end{proof}

Similar to \eqref{eq_control_finite}, there exists a sufficient linear encoding for \eqref{eq_control_infinite}. The feasibility of the linear program implies satisfiability of the contract.

\begin{remark} \label{rmrk_RCI}
We can simplify Theorem \ref{thrm_single_RCI} by assuming the variables $E = 0_{n \times p}$ and $\beta =0$,  as a result, there is no need for constraint (\ref{Econ}). However, it adds to conservativeness and may lead to infeasibility.
\end{remark}

\section{Composition of Parametric Assume-Guarantee Contracts}
\label{sec_AGR}

Now, we shift our focus back to the network of coupled systems in \eqref{subsystems_LTV} and provide the first step toward solutions for Problem \ref{problem_viable} and \ref{Problem_RCI}. Two ideas are presented in this section: (i) we decouple subsystems by viewing all the coupling effects of other subsystems as disturbances. (ii) we make the contract sets parametric, hence the disturbance sets corresponding to couplings also become parametric to search over. The technical details are provided as follows.

\subsection{Composition Correctness}
Unlike the case in the single system where guarantees were obtained from given assumptions using the controller synthesis program, it is much more complicated in the the case of dynamically coupled systems. Because the guarantee of one subsystem affects the assumptions of other subsystems as a result of looking at the coupling effect as a disturbance. We break this circularity by treating the coupling terms in system \eqref{subsystems_LTV} as a disturbance: 
\begin{multline} \label{idea}
     x_i(t+1) = A_{ii}(t)x_i(t) + B_{ii}(t) u_i(t) + \\ \underbrace{\sum_{j \ne i} A_{ij}(t)x_j(t) + \sum_{j \ne i} B_{ij}(t)u_j(t)+ d_i(t)}_{d_i^{aug}(t)},
\end{multline}
where $d_i^{aug}(t) \in D^{aug}_i(t)$ is the augmented disturbance. 

Using this idea and knowing that the assume-guarantee contracts are common knowledge among all subsystems, there is no need for communication between subsystems, facilitating fully decentralized control policies. However, we must first ensure that the controller at every subsystem is correctly designed for the disturbance it expects, which leads us to define a criterion for correctness of a set of assume-guarantee contracts. 

\begin{definition}[Composition Correctness]
Consider a set of valid assume-guarantee contracts $\mathcal{C}_i = (\mathcal{A}_i,\mathcal{G}_i)$, where 
\begin{itemize}
\item $\mathcal{A}_i=\{W_i(t), t \in \mathbb{N}_h\}$ (finite horizon) or $\mathcal{A}_i=W_i$ (infinite horizon);
\item $\mathcal{G}_i=\{(\mathcal{X}_i(t),\mathcal{U}_i(t)), t \in \mathbb{N}_h\}$ (finite horizon) or $\mathcal{G}_i=(\mathcal{X}_i,\mathcal{U}_i)$ (infinite horizon).
\end{itemize}
Then the composition is correct if the following relation holds:
\begin{equation}
D_i^{aug}(t) \subseteq W_i(t)   , \forall i \in \mathcal{I} , \forall t \in \mathbb{N}_{h-1} \text{~(Finite horizon)}
\end{equation}
\begin{equation}
D_i^{aug} \subseteq W_i   , \forall i \in \mathcal{I} , \text{~(Infinite horizon)}
\end{equation}
where 
\begin{equation}
\label{eq_d_aug}
    D_i^{aug}(t)  := \bigoplus_{j \ne i} A_{ij}(t)\mathcal{X}_j(t) \oplus \bigoplus_{j \ne i} B_{ij}(t)\mathcal{U}_j(t) \oplus D_i(t),
\end{equation}
(for infinite horizon, we just drop $(t)$ from \eqref{eq_d_aug}).
\end{definition}


The correctness criterion has a Boolean answer, stating whether the composition of contracts is correct or not. However, we desire a function that describes how far the contracts are from correctness. The following ``potential function'' exactly does that, by dedicating a score to a set of contracts:
\begin{definition}[Potential Function]
Given a set of contracts $\mathcal{C} = \{\mathcal{C}_i |  i \in \mathcal{I}\}$, its potential function is
\begin{equation}
    \mathcal{V}(\mathcal{C}) = \sum_{i \in \mathcal{I}} \mathcal{V}_i(\mathcal{C}),
\end{equation}
where $\mathcal{V}_i(\mathcal{C})$ is defined as follows:
\begin{equation}
    \mathcal{V}_i(\mathcal{C}) := \sum_{t\in \mathbb{N}_{h-1}} d_{DH}(W_i(t) , D_i^{aug}(t)) \text{~ (Finite Horizon)}.
\end{equation}
\begin{equation}
    \mathcal{V}_i(\mathcal{C}) :=  d_{DH}(W_i , D_i^{aug}) \text{~ (Infinite Horizon)}.
\end{equation}
\end{definition}The potential function is sum of the directed Hausdorff distances between the assumption set and the augmented disturbance set, which shows how much the augmented disturbance set is outside of the assumption set. 
When the potential function is zero, then the contracts composed correctly, which means each system is assuming a larger set of disturbances from the actual one happening.

\subsection{Parametric Contracts}
First, we fix sets $\mathcal{X}_i(t)$ and $ \mathcal{U}_i(t)$ in the following form:
\begin{equation}
    \mathcal{X}_i(t) =  \mathcal{Z}(\bar{c}^x_i(t) , C^x_i(t)),
\end{equation}
\begin{equation}
     \mathcal{U}_i(t) =  \mathcal{Z}(\bar{c}^u_i(t) , C^u_i(t)),
\end{equation}
for all $i \in \mathcal{I}$ and $t \in \mathbb{N}_h$, where $\bar{c}^x_i(t) \in \mathbb{R}^{n_i}$, $\bar{c}^u_i(t) \in \mathbb{R}^{m_i}$ and matrices $C^x_i(t) \in \mathbb{R}^{n_i \times \zeta^x_i(t)}$ and $C^u_i(t) \in \mathbb{R}^{m_i \times \zeta^u_i(t)}$. 

Now we introduce parameters $\alpha=\{\alpha^x_i,\alpha^u_i\}_{i \in \mathcal{I}}$ and sets $\mathcal{X}_i(t,\alpha^x_i(t)) \subseteq X_i(t)$ and $\mathcal{U}_i(t,\alpha^u_i(t)) \subseteq U_i(t)$  which are defined as follows:
\begin{subequations}
\begin{equation}
    \mathcal{X}_i(t,\alpha^x_i(t)) := \mathcal{Z}(\bar{c}^x_i(t) , C^x_i(t) \text{Diag}(\alpha^x_i(t))),
\end{equation}
\begin{equation}
    \mathcal{U}_i(t,\alpha^u_i(t)) := \mathcal{Z}(\bar{c}^u_i(t) , C^u_i(t) \text{Diag}(\alpha^u_i(t))),
\end{equation}
\end{subequations}
where $\alpha^x_i(t) \in \mathbb{R}^{\zeta^x_i(t)}$ and $\alpha^u_i(t) \in \mathbb{R}^{\zeta^u_i(t)}$. Basically we multiply each generator of zonotopes $\mathcal{X}_i(t)$ and $\mathcal{U}_i(t)$ by a scalar. 

So far the missing ingredient is the viable sets and contract satisfiability. 
Now we are in the position to combine parameterization with contract satisfiability. 
We bring the notation from Section \ref{sec_AG}. 

The parametric assumption sets are defined as follows: 
\begin{multline}
    W_i(\alpha ,t) := \bigoplus_{j \ne i}[ A_{ij}(t)\mathcal{X}_j(t,\alpha^x_j(t))  \oplus B_{ij}(t)\mathcal{U}_j(t,\alpha^u_j(t))] \\ \oplus D_i(t)   = \mathcal{Z}( \bar{d}_i^{aug}(t) , D_i^{aug}(t)).
\end{multline}
The parameteric guarantees are $\Omega_i(t), \Theta_i(t), t \in \mathbb{N}_{h-1}, i \in  \mathcal{I}$, with all the encoding from programs in Theorem \ref{Thrm_viable set_single} (or Theorem \ref{thrm_single_RCI} for infinite-time, with the drop of $(t)$). We bring all of them into the definition of parametric potential function, defined in the next section. 
\subsection{Parametric Potential Function}
Due to long equations, we provide the encoding only for the finite horizon case with noting that obtaining the infinite horizon case is similar. 
\begin{definition}[Parametric Potential Function]
The parametric potential function is:
\begin{equation}
\label{eq_v_alpha}
    \mathcal{V}(\alpha) = \sum_{i \in \mathcal{I}} \mathcal{V}_i(\alpha),
\end{equation}
where
\begin{multline}
\label{eq_v_i_alpha}
        \mathcal{V}_i(\alpha) := \sum_{t\in \mathbb{N}_{h-1}}  d_{DH}(\mathcal{X}_i(t,\alpha^x_i(t)) ,  \Omega_i(t) ) \\ +  d_{DH}(\mathcal{U}_i(t,\alpha^u_i(t)) ,  \Theta_i(t) ),
\end{multline}
and $\Omega_i(t), \Theta_i(t), t \in \mathbb{N}_{h-1}, i \in \forall \mathcal{I}$ are defined as in Theorem \ref{Thrm_viable set_single} (or Theorem \ref{thrm_single_RCI} for infinite-time, with the drop of $(t)$).
\end{definition}

Note that the parametric potential function is zero when
\begin{equation} \label{corectness_prime}
    \Omega_i(t) \subseteq \mathcal{X}_i(t,\alpha^x_i(t)) \wedge \Theta_i(t) \subseteq \mathcal{U}_i(t,\alpha^u_i(t)),
\end{equation}

Now we need a convex encoding of $V(\alpha)$. This comes at a small price of conservativeness due to the following Lemma, which is a modified version of Lemma \ref{sadra_zon_containment} that is useful in the subsequent sections. 
\begin{lemma}[Weighted Zonotope Containment]
The relation $\mathcal{Z}(\bar{c}_1,G_1) \subseteq \mathcal{Z}(\bar{c}_2,G_2\text{Diag}(\alpha))$ where $\alpha \in \mathbb{R}^s, \alpha>0$, and $s$ is the number of columns in $G_2$ holds, if the conditions in Lemma \ref{sadra_zon_containment} hold while constraint (\ref{con_zon_containment}) changes to:
\begin{equation} \label{con_zon_addmod}
     || [\Gamma,\gamma] ||_\infty \leq \alpha,
\end{equation}
\end{lemma}
\begin{proof}
Using constraints (\ref{zon_containment}) for $\mathcal{Z}(\bar{c}_1,G_1) \subseteq \mathcal{Z}(\bar{c}_2,G_2\text{Diag}(\alpha))$, we reach to:
\begin{equation}
G_1 = G_2 \text{Diag}(\alpha) \Gamma ,  c_2 - c_1 = G_2 \text{Diag}(\alpha) \gamma ,   || [\Gamma,\gamma] ||_\infty \leq 1
\end{equation}
We can replace $\text{Diag}(\alpha) \Gamma$ and $\text{Diag}(\alpha) \gamma$ with $\Gamma ^{new}$ and $\gamma^{new}$, respectively. So we have:
\begin{subequations}
\begin{equation}
G_1 = G_2 \Gamma^{new},
\end{equation}
\begin{equation}
c_2 - c_1 = G_2 \gamma^{new},
\end{equation}
\begin{equation} \label{zon_con_mod}
|| [\text{Diag}(\alpha^{-1}), \Gamma^{new},\text{Diag}(\alpha^{-1})\gamma^{new}] ||_\infty \leq 1
\end{equation}
\end{subequations}
where $\alpha^{-1}$ is element-wise. In  (\ref{zon_con_mod}), each row of matrix $[\Gamma^{new},\gamma^{new}]$ is divided by the corresponding element in vector $\alpha$. Because all the elements of $\alpha$ are positive, we can multiply the inequality by $\text{Diag}(\alpha)$ and have:
\begin{equation}
    || [\Gamma^{new},\gamma^{new}] ||_\infty \leq \alpha
\end{equation}
\end{proof}

The optimization problem for $V_i(\alpha)$ in \eqref{eq_v_i_alpha} is: 
\begin{subequations} \label{LTV_guarantees}
\begin{align}
    & \mathcal{V}_i(\alpha) = \underset{\mathrm{x}^i,T^i,\mathrm{u}^i,M^i , d_t^x , d_t^u}{\text{min}}
    \begin{aligned}[t]
       & \sum_{t \in \mathbb{N}_{h-1}} d^x_t + d^u_t 
    \end{aligned} \notag \\
    &\text{subject to} \notag \\
    & [A_{ii}(t)T^i_t+ B_{ii}(t) M^i_t , D_i^{aug}(t) ] =  [T^i_{t+1}], \hspace{2 mm} \forall t \in \mathbb{N}_{h-1} \label{eq:a} \\
    & A_{ii}(t)\bar{\mathrm{x}}^i_t + B_{ii}(t) \bar{\mathrm{u}}^i_t + \bar{d}_i^{aug}(t) = \bar{\mathrm{x}}^i_{t+1},  \hspace{2 mm} \forall t \in \mathbb{N}_{h-1} \label{eq:b}\\
    & \mathcal{Z}( \bar{d}_i^{aug}(t) , D_i^{aug}(t))  =  \bigoplus_{j \ne i}[A_{ij}(t)\mathcal{X}_j(t,\alpha^x_j(t)) \label{eq:c} \\ & \hspace{2 cm} \oplus B_{ij}(t)\mathcal{U}_j(t,\alpha^u_j(t))] \oplus D_i(t), \forall t \in \mathbb{N}_{h-1}\notag \\
    & \mathcal{Z}(\bar{\mathrm{x}}^i_t,T^i_t) \subseteq \ \mathcal{X}_i(t,\alpha^x_i(t)) \oplus \mathcal{Z}(0,d^x_t I_{n_i}) , \hspace{2 mm} \forall t \in \mathbb{N}_{h-1} \label{eq:d}\\
    & \mathcal{Z}(\bar{\mathrm{u}}^i_t,M^i_t) \subseteq \ \mathcal{U}_i(t,\alpha^u_i(t)) \oplus \mathcal{Z}(0,d^u_t I_{m_i}) , \hspace{1 mm} \forall t \in \mathbb{N}_{h-1} \label{eq:e}\\
    & \mathcal{Z}(\bar{\mathrm{x}}^i_h,T^i_h) \subseteq X_i(h),  \label{eq:f}\\
    &  0 \leq \alpha_i^{x}(t) \leq \alpha_i^{max,x}(t) , \hspace{2 mm} \forall t \in \mathbb{N}_{h-1} \label{eq:g}\\
    & 0 \leq \alpha_i^{u}(t) \leq \alpha_i^{max,u}(t) , \hspace{2 mm} \forall t \in \mathbb{N}_{h-1} \label{eq:h} .
\end{align}
\end{subequations}
Where $d^x_t$ and $d^u_t$ are scalars and $\mathrm{x}^i,T^i,\mathrm{u}^i$, and $M^i$ are sets containing all the $\mathrm{x}^i_t,T^i_t,\mathrm{u}^i_t$, and $M^i_t$, respectively. Constraints (\ref{eq:a}) and (\ref{eq:b}) come from Theorem \ref{Thrm_viable set_single} which enforce viability conditions. The constraint (\ref{eq:c}) is the paramterized assumption and (\ref{eq:f}) is for forcing the state of the final step to be inside the last admissible set, while the upper bounds $\alpha_i^{max,x}(t)$ and $\alpha_i^{max,u}(t)$, in constraints (\ref{eq:g}) and (\ref{eq:h}) (element-wise inequalities), play the same role for state and control input in other time-steps and ensure the validity of the contracts. But, they need to be driven beforehand, such that:
\begin{subequations}
\begin{equation}
     \mathcal{Z}(\bar{c}^x_i(t) , C^x_i(t)\text{Diag}(\alpha_i^{max,x}(t))) \subseteq X_i(t)
\end{equation}
\begin{equation}
     \mathcal{Z}(\bar{c}^u_i(t) , C^u_i(t)\text{Diag}(\alpha_i^{max,u}(t))) \subseteq U_i(t)
\end{equation}
\end{subequations}

Constraints (\ref{eq:d}) and (\ref{eq:e}) and the objective function are for computing the directed Hausdorff distances and we borrow them from \cite{sadraddini2019linear}. 
The following theorem is the main result of this section.
\begin{theorem}[Convexity of parametric potential function]
Using our paramterization and linear encoding with containment, the parametric potential function is convex. And the set of correct parameters (level set at zero) is also a convex set.
\end{theorem}
\begin{proof}
As shown in (\ref{LTV_guarantees}), each $\mathcal{V}_i(\alpha)$ is formulated in a linear program, which implies that each $\mathcal{V}_i(\alpha)$ is a convex function and because the summation of convex functions remains convex, $\mathcal{V}(\alpha)$ is also convex. Moreover, we know that the level set of a convex function is a convex set, so the correct parameters compose a convex set.
\end{proof}


\section{Compositional Synthesis and Computations}
\label{sec_gradient}
Two methods are offered to address problem (\ref{problem_viable}) and (\ref{Problem_RCI}) with the help of parametric assume-guarantee contracts. The first one offers a single centralized optimization to find decentralized viable sets. The second method does the same task, but with a compositional approach.

\subsection{Single Convex Program} \label{centralized_approach}
This is a centralized synthesis method that gives the contracts and a set of decentralized viable sets for each subsystem at the same time. Having each subsystem in the form (\ref{idea}), where $d_i^{aug}(t)$ belongs to
\begin{equation}
    d_i^{aug}(t) \in W_i(\alpha , t),
\end{equation}
with the help of Theorem \ref{Thrm_viable set_single} and (\ref{corectness_prime}), the following centralized optimization is offered for a given $k \in \mathbb{N}$ (in practice, start from an arbitrary initial $k$ and increase it until feasibility is achieved) :

\begin{subequations} \label{LTV_optmz}
\begin{align}
    & \Omega , \Theta = \underset{\mathrm{x}^i_t,T^i_t,\mathrm{u}^i_t,M^i_t,\alpha}{\text{argmin}}
    \begin{aligned}[t]
       & \sum_{t \in \mathbb{N}_{h-1}}\sum_{i\in \mathcal{I}}{\text{sum}(\alpha_i^x(t))}
    \end{aligned} \notag \\
    &\text{subject to} \notag \\
    & [A_{ii}(t)T^i_t+ B_{ii}(t) M^i_t , D_i^{aug}(t) ] =  [T^i_{t+1}] , \hspace{2 mm} \forall t\in \mathbb{N}_{h-1} , \forall i \in \mathcal{I}  \label{eq2:a} \\
    & A_{ii}(t)\bar{\mathrm{x}}^i_t + B_{ii}(t) \bar{\mathrm{u}}^i_t + \bar{d}_i^{aug}(t) = \bar{\mathrm{x}}^i_{t+1} , \hspace{4 mm} \forall t\in \mathbb{N}_{h-1} ,  \forall i \in \mathcal{I} \label{eq2:b}\\
    & \mathcal{Z}(\bar{\mathrm{x}}_t^i,T^i_t) \subseteq X_i(t) , \hspace{4 mm} \forall t\in \mathbb{N}_{h} , \forall i \in \mathcal{I}        \label{eq2:c}\\
    &  \mathcal{Z}(\bar{\mathrm{u}}^i_t,M^i_t) \subseteq U_i(t), \hspace{4 mm} \forall t\in \mathbb{N}_{h-1} , \forall i \in \mathcal{I}  \label{eq2:d}\\
    & \mathcal{Z}( \bar{d}_i^{aug}(t) , D_i^{aug}(t))  =  \bigoplus_{j \ne i}[A_{ij}(t)\mathcal{X}_j(t,\alpha^x_j(t)) \label{eq2:e}\\ & \hspace{1.5 cm} \oplus B_{ij}(t)\mathcal{U}_j(t,\alpha^u_j(t))] \oplus D_i(t),  \forall t \in \mathbb{N}_{h-1} , \forall i \in \mathcal{I} \notag \\
    & \mathcal{Z}(\bar{\mathrm{x}}_t^i,T^i_t) \subseteq \mathcal{X}_i(t,\alpha^x_i(t)), \hspace{4 mm}\forall t\in \mathbb{N}_{h-1} , \forall i \in \mathcal{I}     \label{eq2:f}\\
    & \mathcal{Z}(\bar{\mathrm{u}}^i_t,M^i_t) \subseteq \mathcal{U}_i(t,\alpha^u_i(t)), \hspace{4 mm} \forall t\in \mathbb{N}_{h-1} , \forall i \in \mathcal{I}   \label{eq2:g}\\
    & \alpha^x_i(t),\alpha^u_i(t) \geq 0,  \hspace{4 mm} \forall t\in \mathbb{N}_{h-1} ,  \forall i \in \mathcal{I} \label{eq2:h}.
\end{align}
\end{subequations}
Where $\Omega = \{ \Omega_i(t)| \Omega_i(t) = \mathcal{Z}(\mathrm{x}^i_t,T^i_t) ,\forall t\in \mathbb{N}_{h} , \forall i \in \mathcal{I} \}$ and $\Theta = \{ \Theta_i(t)| \Theta_i(t) = \mathcal{Z}(\mathrm{u}^i_t,M^i_t) ,\forall t\in \mathbb{N}_{h-1} , \forall i \in \mathcal{I} \}$. The constraints (\ref{eq2:a}), (\ref{eq2:b}), (\ref{eq2:c}), and (\ref{eq2:d}) imply the viable sets constraints from Theorem \ref{Thrm_viable set_single}. The constraints (\ref{eq2:f}) and (\ref{eq2:g}) are coming from the correctness criterion (\ref{corectness_prime}). The arbitrary sets $\mathcal{X}_i(t)$ and $\mathcal{U}_i(t)$ can be determined by a prior knowledge of the system (e.g. can be the viable set of the subsystem neglecting the coupling effects) or they can simply be the whole admissible state space $X_i(t)$ and control input $U_i(t)$, respectively. In that case, there is no need for the constraints (\ref{eq2:c}) and (\ref{eq2:d}), instead we need to add $\mathcal{Z}(\bar{\mathrm{x}}^i_h,T^i_h) \subseteq X_i(h) , \forall i \in \mathcal{I}$  to the constraints and the constraint (\ref{eq2:h}) changes to
\begin{equation}
  0 \leq  \alpha^x_i(t),\alpha^u_i(t) \leq 1,  \hspace{4 mm} \forall t\in \mathbb{N}_{h-1} ,  \forall i \in \mathcal{I}.
\end{equation}
The objective function is the summation of all the elements of $\alpha^x_i$ over all subsystems (and time steps), which is a heuristic method for minimizing the volume of the viable sets. Note that this method is sound, because the correctness criterion is enforced in the process by using zonotope containment constraints. As a result, any output of the approach is correct-by-construction.

\subsection{Compositional Approach} \label{compos_method}

The centralized method in the previous section is not scalable to large-scale systems, although the implementation (on line mode) is decentralized, but solving one large linear program is still problematic. Two main causes of this problem are: (i) the large number of variables and constraints in a single optimization problem (ii) the order of zonotope $\mathcal{Z}( \bar{d}_i^{aug}(t) , D_i^{aug}(t))$ is very large when the number of subsystems is large. To address (i), we propose a compositional method that, at each optimization, just deals with one subsystem. Also for solving (ii), we use zonotope order reduction methods to over-approximate the disturbance set with a zonotope with smaller order \cite{Kopetzki2018}, \cite{Yang2018}. For this paper, we have used \textit{boxing method} \cite{Kiihn1998}, \cite{C.Combastel} which is a well-known zonotope order reduction method. Using \textit{boxing method} with the desired order $o_i(t)$, the new disturbance set would be:
\begin{equation}
    \mathcal{Z}( \bar{d}_i^{red}(t) , D_i^{red}(t)) \xleftarrow{\text{order reduction}} \mathcal{Z}( \bar{d}_i^{aug}(t) , D_i^{aug}(t)),
\end{equation}
where $\bar{d}_i^{red}(t) \in \mathbb{R}^{n_i} $ and $D_i^{red}(t) \in \mathbb{R}^{n_i \times (o_i(t)n_i)}$. One of the main contribution of this paper is proposing a convex potential function, where each subsystem has its own cost and the potential function is the summation of all costs, which enables us to solve the following optimization problem in a distributed manner:

\begin{equation} \label{potential_optmz}
\mathcal{V}^* = \underset{\alpha}{\text{min   }} \sum_{i \in \mathcal{I}}\mathcal{V}_i(\alpha) 
\end{equation}
If $\mathcal{V}^* =0$, then the contracts match each other and it is allowable to use $\alpha^*$ and find the viable sets. Otherwise ($\mathcal{V}^* > 0$), the method has failed to find a correct set of contracts and we need to increase $k$ or $o_i(t)$ and try again. Note that each $\mathcal{V}_i(\alpha)$ can be computed from (\ref{LTV_guarantees}) while the constraints (\ref{eq:g}) and (\ref{eq:h}) are removed for a fixed $k \in \mathbb{N}$. The optimization problem (\ref{potential_optmz}) is a convex optimization problem, which can be solved by gradient descent:  
\begin{equation} \label{gradient_descent}
    \alpha \leftarrow \alpha - \delta  \nabla \mathcal{V}(\alpha),
\end{equation}
where $\delta$ is the step size. The gradient of $\mathcal{V}(\alpha)$ is equal to:
\begin{equation}
    \nabla \mathcal{V}(\alpha) = \sum_{i \in \mathcal{I}} \nabla \mathcal{V}_i(\alpha).
\end{equation}
It is well-known that the dual variable of a constraint shows the gradient with respect to the right hand side of that constraint \cite{bertsimas1997introduction}. The optimization problem (\ref{LTV_guarantees}) are formulated in a way that all the elements of $\alpha$ locates in the right hand side of the constraints, so that, the dual variable of the corresponding constraint will give the gradient with respect to the right hand side of that constraint. Then, by using the chain rule, we can compute $\nabla \mathcal{V}_i(\alpha)$ with respect to the elements of $\alpha$.

It is important to note that a set of contract can satisfy correctness criterion (\ref{corectness_prime}), but be not inside the accessible state space or control input. As explained before, That is the reason for existence of upper bounds in the constraints (\ref{eq:g}) and (\ref{eq:h}), but the $\alpha$ can jump outside of its feasible region because of discrete jumps in the process of gradient descent method in (\ref{gradient_descent}). To fix this problem, whenever it happens we need to map the current $\alpha$ to its feasible region.

\section{Case Studies}
\label{sec_examples}
We present three case studies as follows:
\begin{itemize}
    \item 
The first one is an illustrative example to show how the compositional method works and how the convex potential function looks like for a simple example. 
    \item
The second case study shows decentralized finite-time viable sets for an LTV system using the single convex program.
    \item
The third one Benchmarks the scalability of the proposed compositional method with respect to three existing methods.
\end{itemize}
We used a MacBook Pro 2.6 GHz Intel Core i7 and Gurobi optimization solver \cite{gurobi} to do the computations.
\subsection*{Case Study 1}
Consider a time-invariant system in the form (\ref{subsystems_LTV}). The aggregated matrix $A$ which contains all the $A_{ij}$s is as follows:

\begin{figure}[t]
  \centering
  \includegraphics[scale=0.45]{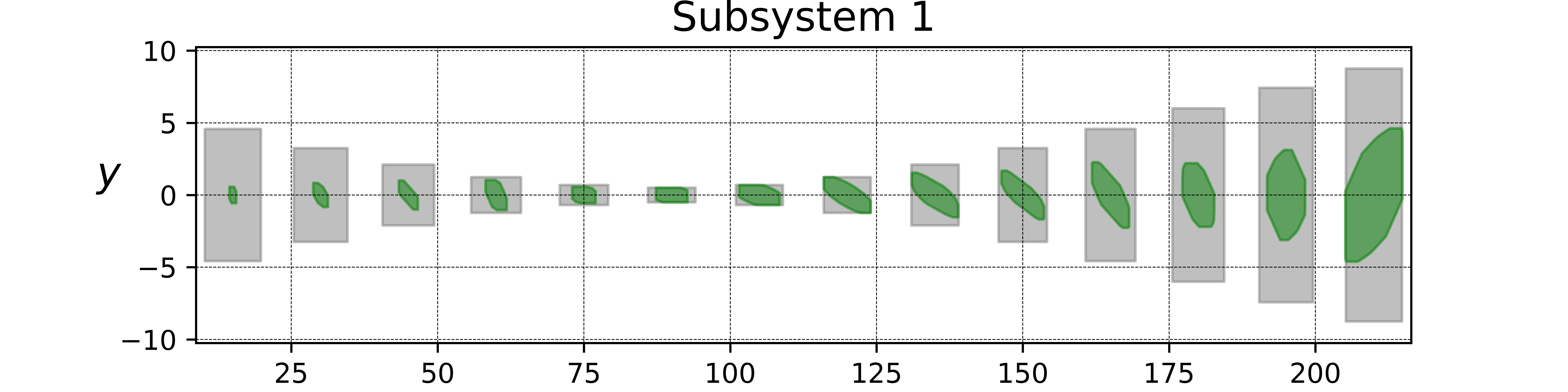}
  \includegraphics[scale=0.45]{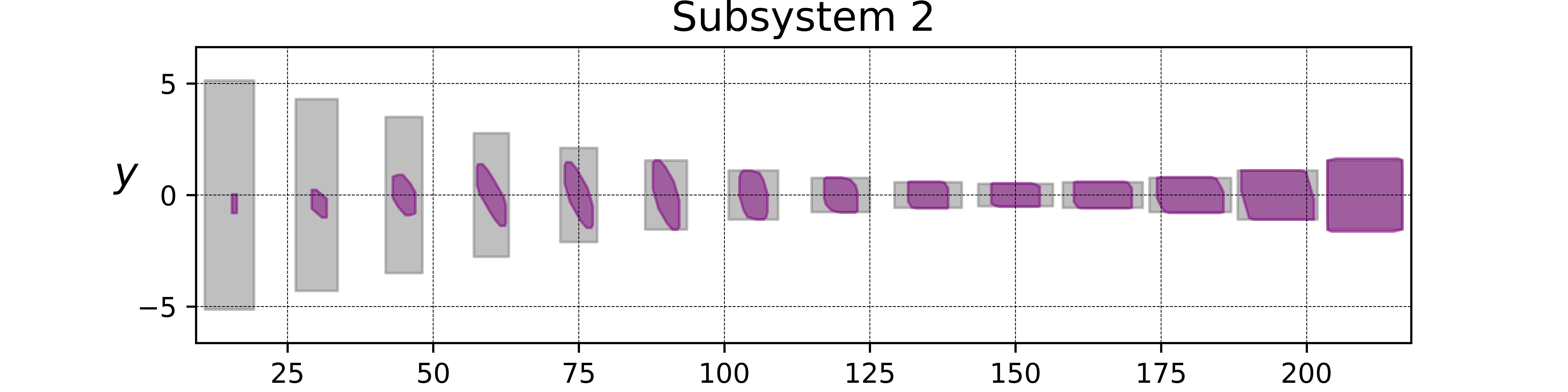}
  \includegraphics[scale=0.45]{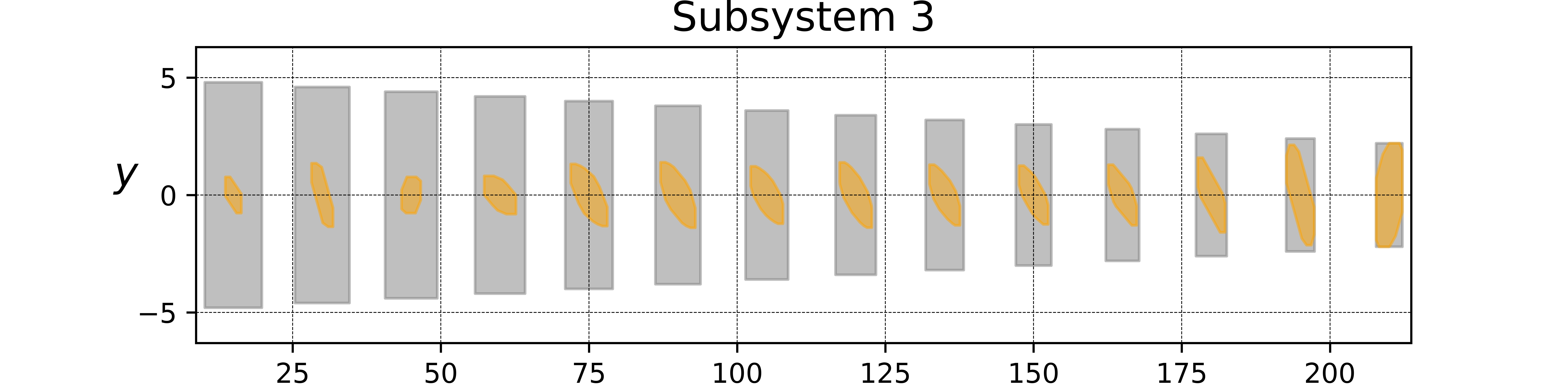}
  \caption{It shows decentralized viable sets for each subsystem. To prevent overlapping of the viable sets, each viable set is moved 15 units right with respect to the previous viable set. Each gray area shows $X_i(t)$ (the bound over the state) and all viable sets are correctly inside their corresponding state bounds.}
  \label{Case_study_2}
\end{figure}

\begin{equation*}
A = \left[\begin{array}{c c| c c |c c}
	1 & 1.1 & 0.1 & 0.01 & 0.8 & 0.1 \\
0 & 1  & 0.1 & 0.01 & 0.8 & 0.1 \\
\hline
0.1 & 0.01 & 1 & 1.1 & 0.4 & 0.01 \\
0.1 & 0.01 & 0 & 1 & 0.4 & 0.01 \\
\hline
0.02 & 0.0001 & 0.01 & 0.0001 & 1 & 1.1 \\
0.02 & 0.0001 & 0.01 & 0.0001 & 1 & 1 
\end{array}\right],
\end{equation*} 
where $n_i =2$ for $i \in \{1,2,3 \}$ and $A_{ij}$ is a square matrix inside $A$, such that $i$ is the number of row and $j$ is the number of column of the $2 \times 2$ square matrix from top left), which means that there are three coupled subsystems and the couplings are just on the states. And 

\begin{equation*}
     B_{ii} = \begin{bmatrix} 
0 \\
0.1
\end{bmatrix},
G^x_i =      \begin{bmatrix} 
1 & 0 \\
0 & 1
\end{bmatrix},
\end{equation*}
and the origin is the center of all zonotops. The goal is to find infinite-time decentralized viable sets for each subsystem (problem \ref{Problem_RCI}). We used the proposed compositional method in section \ref{compos_method}. Note that because the problem is time-invariant, we have to drop $t$ in optimization (\ref{LTV_guarantees}) and replace the constraints that come from Theorem \ref{Thrm_viable set_single} by the constraints in Theorem \ref{thrm_single_RCI}. The dimension of $\alpha^x=  [\alpha^x_1,\alpha^x_2,\alpha^x_3]$ is $6$ (two for each subsystem) and the first and second element of each $\alpha^x_i$ is shown by $\alpha^x_i[1]$ and $\alpha^x_i[2]$, respectively. Since the subsystems are not coupled by their control input, there is no need to define $\alpha^u_i$. The results are shown in the Fig. \ref{figurelabe2}, where the top figure shows the projection of the zero level set of the parametric potential function in $\alpha^x_1[1]$ and $\alpha^x_2[1]$ plane. The incorrect area of parameters is shown in red and the compositionally correct region is shaded in green. The trajectory shows the updates of $\alpha^x$ derived from gradient descent (\ref{gradient_descent}), which starts from an initial $\alpha^x$ and ends in area of correct parameters and the arrows show the direction of gradients of the parametric potential function for each subsystem.

\subsection*{Case Study 2}
In this example, we demonstrated decentralized finite-time viable sets for three coupled LTV systems in the form (\ref{subsystems_LTV}) with the following characteristics:
\begin{multline*}
    A = \left[\begin{array}{c c| c c |c c}
1  & 1.1  & 0.002 & 0.002 & 0 & 0 \\
0 & 1 & 0.002 & 0.002 & 0 & 0 \\
\hline
0.002 & 0.002 & 1 & 1.1 & 0.002 & 0.002 \\
0.002 & 0.002 & 0 & 1 & 0.002 & 0.002 \\
\hline
0 & 0 & 0.002 & 0.002 & 1 & 1.1 \\
0 & 0 & 0.002 & 0.002 & 0 & 1
\end{array}\right],\\
B_{ii}(t) = \begin{bmatrix} 
0\\
0.1 
\end{bmatrix}
,   B_{ij}(t) = \begin{bmatrix}
0 \\
0
\end{bmatrix},
   D_i(t) = \mathcal{Z}(0 , \begin{bmatrix} 
0.4 & 0 \\
0 & 0.4 
\end{bmatrix}) ,\\
   U_i(t) = \mathcal{Z}(0 ,[10]), 
   X_1(t) = \mathcal{Z}(0 , \begin{bmatrix}
5 -  \dfrac{\pi t }{15} & 0 \\
0 & 6 - \dfrac{11 \pi t}{24}
\end{bmatrix}), \\
   X_2(t) = \mathcal{Z}(0 , \begin{bmatrix}
5 - 2 * \sin{ \dfrac{\pi t}{8}} & 0 \\
0 & 6 - 5.5 * \sin{ \dfrac{\pi t}{20}}
\end{bmatrix}),\\
X_3(t) = \mathcal{Z}(0 , \begin{bmatrix}
5 - \dfrac{t}{5} &  0 \\ 
0. &  5 - \dfrac{t}{5}
\end{bmatrix}),
\end{multline*}
where $t \in \mathbb{N}_{15}$, $n_i =2$, $m_i = 1$, for $i \in \{1,2,3\}$. $A$ is the aggregated matrix of $A_{ij}$s (like case study 1). The resulted decentralized viable sets are shown in the Fig. \ref{Case_study_2}. Note that, for each subsystem, the viable set of the first step is a point since the proposed objective function minimizes the size of the viable sets.

\subsection*{Case Study 3}
\label{case_table}
\begin{table*}[t]
\caption{Case Study \ref{case_table}: Synthesis Times in Seconds}
\resizebox{0.9\textwidth}{!}{
\begin{tabular}{|c|c|c|c|c|}
\hline
\begin{tabular}[c]{@{}c@{}}total dimension\\ of the state space\\ (= $2 \times$ number of subsystems)\end{tabular} & $\lambda$ (Coupling Parameter)       & \begin{tabular}[c]{@{}c@{}}Centralized Optimization \\ of \\ Centralized Controllers \end{tabular} & \begin{tabular}[c]{@{}c@{}}Centralized Optimization \\ of \\ Decentralized Controllers\end{tabular} & \begin{tabular}[c]{@{}c@{}} Compositional Synthesis \\ of \\ Decentralized Controllers  \end{tabular} \\ \hline
10                                                                                                      & 1       & 1.11                                                                                      & 0.87                                                                             & 0.011                                                                              \\ \hline
20                                                                                                      & 0.1     & 14.58                                                                                     & 6.75                                                                             & 0.023                                                                              \\ \hline
40                                                                                                      & 0.1     & 211.72                                                                                    & 54.95                                                                            & 0.048                                                                              \\ \hline
60                                                                                                      & 0.1     & 1046.69                                                                                   & 192.10                                                                           & 0.64                                                                               \\ \hline
80                                                                                                      & 0.1     & time out                                                                                  & 472.49                                                                           & 1.28                                                                               \\ \hline
100                                                                                                     & 0.1     & time out                                                                                  & 961.23                                                                           & 3.60                                                                               \\ \hline
200                                                                                                     & 0.05    & time out                                                                                  & time out                                                                         & 7.49                                                                               \\ \hline
400                                                                                                     & 0.05    & time out                                                                                  & time out                                                                         & 56.12                                                                              \\ \hline
500                                                                                                     & 0.05    & time out                                                                                  & time out                                                                         & 5.38                                                                               \\ \hline
1000                                                                                                    & 0.01    & time out                                                                                  & time out                                                                         & 26.63                                                                              \\ \hline
2000                                                                                                    & 0.001   & time out                                                                                  & time out                                                                         & 11.94                                                                              \\ \hline
4000                                                                                                    & 0.001   & time out                                                                                  & time out                                                                         & 38.83                                                                              \\ \hline
10000                                                                                                   & 0.0001  & time out                                                                                  & time out                                                                         & 90.31                                                                              \\ \hline
20000                                                                                                   & 0.00001 & time out                                                                                  & time out                                                                         & 217.27                                                                             \\ \hline
\end{tabular} \label{tabel1}
}
\end{table*}
This example is adopted from \cite{Motee2008}, where the authors generated a random network of coupled linear subsystems. They initially scatter random points in a square field with each side 100 units and assign each point to a subsystem. If the Euclidean distance between any two points is less than $10$ units, they are considered as neighbors. The dynamics for each subsystem is:
\begin{equation} \label{eq:ex2}
    x_i^+ = A_{ii}x_i(t) + B_{ii}u_i(t) + d_i(t)+ \sum_{j \neq i}{A_{ij}x_j(t)},
\end{equation}
where $A_{ii}$ is $\begin{bmatrix} 
    1 & 1.2\\
    0 & 1
    \end{bmatrix}$ and $B_{ii}$ is $\begin{bmatrix} 
    0\\
    0.2
    \end{bmatrix} $. If subsystems $i$ and $j$ are not neighbors, $A_{ij}=0 $. Otherwise: 
\begin{equation} \label{eq: lambda}
    A_{ij} = \dfrac{\lambda}{1 + \dist(i,j)}
    \begin{bmatrix} 
    1 & 1\\
    1 & 1
    \end{bmatrix},
\end{equation}
where $\lambda$ is a constant and $\dist(i,j)$ is Euclidean distance between points $i$ and $j$.
The following constraints are imposed on (\ref{eq:ex2}):
\begin{multline}
    x_i(t) \in \mathcal{Z}(0 , \begin{bmatrix} 
10 & 0 & 10\\
0 & 10 & -10 
\end{bmatrix}), u_i(t) \in \mathcal{Z}(0,10I_1),\\ d_i(t) \in \mathcal{Z}(0, 0.2I_2).
\end{multline}
The problem is finding infinite-time contracts for each subsystem. We solve it by three different methods and report the execution times for different sizes of the total state space dimension in Table~\ref{tabel1}. The first method (corresponds to the third column in the table) is the conventional centralized method, which comes from Theorem \ref{thrm_single_RCI} and results in centralized viable sets. As expected before, it could not be applied to large-scale systems. The second approach (corresponds to the fourth column) is our proposed centralized method in section \ref{centralized_approach} with some adjustments for infinite-time contracts. It did a better job than the first one because it has less dense controllers. The third method (corresponding to the last column) is our proposed compositional method which shows great scalability and helps to solve up to 20,000 dim. Note that the reported times for the compositional method are the aggregated time for just solving the optimization problems, excluding the time for building the optimization problem, which heavily depends on the programming interface. Additionally, for the compositional method, all the orders of reduced zonotopes ($o_i(t)$) are fixed to 1. Also, we need to use Remark \ref{rmrk_RCI} to still have linear programs. When the number of subsystems increases, the coupling effects get larger and it may lead to infeasibility. For the sake of getting feasibility all the time, $\lambda$ decreases as the number of subsystems increases.

\section{Conclusion and Future works}
We identified a convex paramterization of assume-guarantee contracts that facilitated  compositional control synthesis of decentralized controllers for large-scale linear systems. The method scales well to very large problems. 

Future work will focus on identifying richer classes of parameterization, and extension to nonlinear systems.

\section{Acknowledgment}
This work was partially supported at Boston University by the NSF under grant IIS-1723995.


\newpage

\bibliographystyle{ACM-Reference-Format}
\bibliography{references_HSCC}

\end{document}